\documentclass[12pt]{extarticle}
\usepackage{extsizes}
\usepackage[reqno]{amsmath} 
\usepackage{amsmath}
\usepackage{amssymb}
\usepackage{amsthm}
\usepackage{amscd}
\usepackage{amsfonts}
\usepackage{graphicx}
\usepackage{dsfont}
\usepackage{fancyhdr}
\usepackage{enumerate}
\usepackage{youngtab}
\usepackage[utf8]{inputenc}
\usepackage{comment}

\usepackage[backend=biber,style=numeric,doi=false,isbn=false,url=false,hyperref=true,backref=true]{biblatex}
\bibliography{qwalk.bib}

\usepackage{subfigure}
\usepackage[margin=1.5in]{geometry}
\usepackage{graphicx}
\usepackage[noend]{algorithmic}
\usepackage{algorithm}
\usepackage{color}
\usepackage{hyperref}
\usepackage{notation}

\theoremstyle{definition}
\newtheorem{corollary}{Corollary}

\newtheorem{lemma}[corollary]{Lemma}

\newtheorem{theorem}[corollary]{Theorem}
\newtheorem{problem}{Problem}
\newtheorem{conjecture}{Conjecture}

\begin{document}

\AtEndDocument{%
  \par
  \medskip
  \begin{tabular}{@{}l@{}}%
    \textsc{Gabriel Coutinho}\\
    \textsc{Dept. of Computer Science}\\
    \textsc{Universidade Federal de Minas Gerais, Brazil}\\
    \textit{E-mail address}: \texttt{gabriel@dcc.ufmg.br} \\ \ \\
    \textsc{Chris Godsil}\\
    \textsc{Dept. of Combinatorics and Optimization}\\
    \textsc{University of Waterloo, Canada}\\
    \textit{E-mail address}: \texttt{cgodsil@uwaterloo.ca} \\ \ \\
    \textsc{Emanuel Juliano}\\
    \textsc{Dept. of Computer Science}\\
    \textsc{Universidade Federal de Minas Gerais, Brazil}\\
    \textit{E-mail address}: \texttt{emanueljulianoms@gmail.com}\\ \ \\
    \textsc{Christopher M.\ van Bommel}\\
    \textsc{Dept. of Mathematics}\\
    \textsc{University of Manitoba, Canada}\\
    \textit{E-mail address}: \texttt{Christopher.vanBommel@umanitoba.ca}
  \end{tabular}}

\title{Quantum walks do not like bridges}
\author{Gabriel Coutinho \and Chris Godsil \and Emanuel Juliano \and Christopher M.\ van Bommel}
\date{\today}
\maketitle
\vspace{-0.8cm}

\begin{abstract}
       We consider graphs with two cut vertices joined by a path with one or two edges, and prove that there can be no quantum perfect state transfer between these vertices, unless the graph has no other vertex. We achieve this result by applying the 1-sum lemma for the characteristic polynomial of graphs, the neutrino identities that relate entries of eigenprojectors and eigenvalues, and variational principles for eigenvalues (Cauchy interlacing, Weyl inequalities and Wielandt minimax principle). We see our result as an intermediate step to broaden the understanding of how connectivity plays a key role in quantum walks, and as further evidence of the conjecture that no tree on four or more vertices admits state transfer. We conclude with some open problems.
\end{abstract}

\begin{center}
\textbf{Keywords}
\end{center}

\textsc{quantum walks; state transfer; graph 1-sum; interlacing}

\section{Introduction}

Let $X$ be a graph, understood to model a network of interacting qubits. Upon certain initial setups for the system, the time evolution is determined by the matrix
\[
	U(t) = \exp(\ii t A),
\]
where $t \in \Rds_+$ and $A = A(X)$, the adjacency matrix of $X$. In this paper we choose to use the bra-ket notation: a vertex $a$ of the graph is represented by a $01$-characteristic vector $\ket a$. The dual functional is denoted by $\bra a$. We say that $X$ admits \textit{perfect state transfer} between $a$ and $b$ at time $t$ if
\[
	\big|\bra{b} U(t) \ket{a} \big| = 1.
\]
For an introduction to the topic we recommend \cite{CoutinhoPhD}.

Quantum perfect state transfer is a desirable phenomenon for several applications in quantum information and yet it is difficult to obtain. Path graphs on $2$ and $3$ vertices admit it, but no other \cite{ChristandlPSTQuantumSpinNet2}, and no other tree is known to achieve it \cite{CoutinhoLiu2}. The infinite families of graphs known to admit state transfer all have an exponential growth compared to the distance between the two vertices involved, while cost constraints in building quantum networks suggest that the desirable configurations should have polynomial growth \cite{KayLimbo}.

Upon allowing for edge weights, it is possible to achieve state transfer on paths, but again, the known families (see for instance \cite{VinetZhedanovHowTo,VinetZhedanovDualHahnPols}) require large weights on the centre of the chain. A question raised in the literature \cite{Casaccino} asked if it was possible to achieve state transfer on a path modulating the weights of loops placed at the extremes of the chain only. In \cite{LippnerPotential}, this was answered in the negative. Our investigation in this paper is related to theirs and in some sense slightly more general: we connect two vertices by a path, and ask if a graph can be used to decorate each end of this chain so that the state transfer happens between the two vertices. We answer this question partially for when the path has one or two edges, also in the negative. We use several standard techniques from linear algebra, some of which not yet used in the context of quantum walks to the best of our knowledge, thus bringing perhaps new inspiration for future research.

In Section \ref{sec:prelim} we state all known results we use in this paper for the convenience of the reader. In Section \ref{sec:strongcut} we show a new result that lays the groundwork for our further analysis. In Sections \ref{sec:p2} and \ref{sec:p3} we prove that state transfer does not happen when the two special vertices and the graph between them induces $P_2$ and $P_3$, respectively in each section. In Section \ref{sec:final} we list open problems and future lines of investigation.

\section{Preliminaries} \label{sec:prelim}

Assume we have a graph $Z$ with two cut vertices $a$ and $b$, just like the figure below.

\begin{figure}[h] 
    \centering
	\includegraphics[scale=0.5]{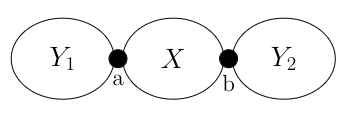}
    \caption{Graph with two cut vertices, called $Z$.} \label{figure1}
\end{figure}

Our goal is to show that if $X$ is $P_2$ or $P_3$, then perfect state transfer does not happen between $a$ and $b$, unless of course $Y_1$ and $Y_2$ are trivial graphs containing only one vertex.

\subsection{State transfer}

Given a graph $X$ on $n$ vertices with adjacency matrix $A$, we assume the spectral decomposition of $A$ is denoted by
\[
	A = \sum_{r = 0}^d \theta_r E_r,
\]
thus we assume there are $d+1$ distinct eigenvalues $\theta_r$, with corresponding eigenprojectors $E_r$. We assume the graph is connected, $\theta_0$ is the largest eigenvalue, and thus $E_0$ is a matrix with positive entries (see \cite[Section 2.2]{BrouwerHaemers}). Then
\[
	U(t) = \exp(\ii t A) = \sum_{r = 0}^d \e^{\ii t \theta_r} E_r,
\]
and it is immediate to verify that, for $a,b \in V(X)$, there is $t$ so that ${|\bra b U(t) \ket a| = 1}$ if and only if there is $\gamma \in \Cds$ with $|\gamma| = 1$ so that $U(t) \ket a = \gamma \ket b$. This equation is equivalent to having, for all $r \in \{0,\cdots, d\}$, 
\[
	\e^{\ii t \theta_r}E_r \ket a = \gamma E_r \ket b,
\]
which is then equivalent to having, simultaneously, for all $r$,
\begin{enumerate}[(a)]
	\item $E_r \ket a = \sigma_r E_r \ket b$, with $\sigma_r \in \{-1,+1\}$, and
	\item whenever $E_r \ket a \neq 0$, then $t (\theta_0 - \theta_r) = k_r \pi$, with $k_r \in \Zds$, and moreover ${k_r \equiv (1-\sigma_r)/2 \pmod 2}$.
\end{enumerate}
Two vertices for which condition (a) holds are called \textit{strongly cospectral}. Note that it implies $\bra a E_r \ket a = \bra b E_r \ket b$ for all $r$, which is the weaker more well known condition that the vertices are cospectral. It is immediate to verify that cospectral vertices satisfy $\bra a A^k \ket a = \bra b A^k \ket b$ for all $k$, and therefore they must have the same degree.

Eigenvalues $\theta_r$ for which $E_r \ket a \neq 0$ are said to belong to the \textit{eigenvalue support} of $a$.
 
Godsil showed that condition (b) above implies that the eigenvalues are either integers or quadratic integers of a special form \cite{GodsilPerfectStateTransfer12}, and from this we obtain the following characterization of perfect state transfer (see for instance \cite[Chapter 2]{CoutinhoPhD}).

\begin{theorem} \label{thm:pstcha}
	Let $X$ be a graph, and let $a,b \in V(X)$. There is perfect state transfer between $a$ and $b$ at time $t$ if and only if all conditions below hold.
	\begin{enumerate}[(a)]
		\item $E_r \ket a = \sigma_r E_r \ket b$, with $\sigma_r \in \{-1,+1\}$.
		\item There is an integer $\alpha$, a square-free positive integer $\Delta$ (possibly equal to 1), so that for all $\theta_r$ in the support of $a$, there is $\beta_r$ giving
			\[
				\theta_r = \frac{\alpha + \beta_r \sqrt{\Delta}}{2}.
			\]
			In particular, because $\theta_r$ is an algebraic integer, it follows that all $\beta_r$ have the same parity as $a$.
		\item There is $g \in \Zds$ so that, for all $\theta_r$ in the support of $a$, $(\beta_0 - \beta_r)/g = k_r$, with $k_r \in \Zds$, and ${k_r \equiv (1-\sigma_r)/2 \pmod 2}$.
	\end{enumerate}
	If the conditions hold, then the positive values of $t$ for which perfect state transfer occurs are precisely the odd multiples of $\pi/(g \sqrt{\Delta})$.
\end{theorem}

\subsection{1-sum lemma}

In this paper, we will investigate perfect state transfer between cut-vertices. Fortunately, there is a very simple recurrence for the characteristic polynomial of a graph in terms of those of some of its subgraphs when a cut-vertex is deleted. This result is likely due to Schwenk (see for instance \cite[Corollary 2b]{SchwenkComputing}). We shall use $\phi(X;t)$ to the denote the characteristic polynomial of the graph $X$ on the variable $t$.

Suppose $X$ and $Y$ are disjoint graphs, and let $Z$ be the graph obtained by identifying a vertex of $X$ with a vertex of $Y$. We say that $Z$ is a $1$-sum of $X$ and $Y$ at the identified vertex.

\begin{lemma} \label{1sum}
	If $Z$ is the $1$-sum of $Y_1$ and $Y_2$ at $b$, then
	\[
		\phi(Z;t) = \phi(Y_1;t)\phi(Y_2\backslash b;t)+\phi(Y_1\backslash b;t)\phi(Y_2;t) - t \phi(Y_1\backslash b;t)\phi(Y_2\backslash b;t).
	\]
\end{lemma}

Because this result is perhaps not so well known, we present its proof (which is different from the original proof in Schwenk's work).

\begin{proof}
	Let $W_{a}(X;t)$ be the walk generating function for the closed walks that start and end at vertex $a$ (thus, the coefficient of $x^k$ counts the number of closed walks that start and end at $a$ after $k$ steps). Note that 
	\[W_{a}(X;t) = \left(\sum_{k \geq 0} A^k x^k \right)_{a,a} = (I - xA)^{-1}_{a,a}.\] 
	From the adjugate expression for the inverse, it follows that
	\begin{align}\label{eq:adju}
		t^{-1}W_{a}(X;t^{-1}) = \frac{\phi(X \backslash a;t)}{\phi(X;t)}.
	\end{align}
	Let now $C_{a}(X;t)$ be the walk generating function for the closed walks that start and end at vertex $a$ but return to $a$ only at the final step. Any walk that starts and ends at $a$ can be decomposed into a walk that starts and ends at $a$, followed by another that starts at $a$ and returns exactly once. Thus
	\[
    	W_{a}(X;t)(1-C_{a}(X;t)) = 1,
	\]
	and therefore
	\[
    	C_{a}(X;t) =1 - W_{a}(X;t)^{-1}.
	\]
	Finally, we have
	\[
    	C_{b}(Z;t) = C_{b}(Y_1;t) + C_{b}(Y_2;t).
	\]
The rest follows from Equation \eqref{eq:adju}.
\end{proof}

\subsection{Neutrino identities} \label{neutrino}

The key to our analysis will be the ability to write the entries of $E_r$ in terms of the characteristic polynomial of vertex deleted subgraphs and the eigenvalues of $A$. For details on what follows below, we refer the reader to \cite[Chapter 4]{GodsilAlgebraicCombinatorics}.

Working with the generating function formalism, we consider
\[
	\sum_{k \geq 0} A^k t^k = (I - tA)^{-1},
\]
which leads to the expression
\begin{align} \label{eq:neutrino0}
	(tI - A)^{-1} = \sum_{r = 0}^d \frac{1}{t-\theta_r} E_r.
\end{align}
By using the adjugate matrix expression for the inverse of a matrix, it follows that
\begin{align} \label{eq:neutrino1}
	\bra a E_r \ket a = \frac{(t - \theta_r)\phi(X \backslash a;t)}{\phi(X;t)} \Bigg|_{t = \theta_r},
\end{align}
where this is to be understood as a way of recovering the coefficient of $(t-\theta_r)^{-1}$ in the expansion of $\phi(X \backslash a;t)/\phi(X;t)$.
	
With a little more work and using a result due to Jacobi, one obtains
\begin{align} \label{eq:neutrino2}
	\bra b E_r \ket a = \frac{(t - \theta_r)\sqrt{\phi(X \backslash a;t)\phi(X \backslash b;t) - \phi(X;t)\phi(X \backslash \{a,b\};t)}}{\phi(X;t)} \Bigg|_{t = \theta_r}.
\end{align}
The square root can be shown to be a polynomial, and it has an expression in terms of path deleted subgraphs. If $\mathcal{P}_{ab}$ is the set of all vertex sets of paths between $a$ and $b$ (inclusive), then it is an exercise to show that
\begin{align} \label{eq:neutrino3}
	\sqrt{\phi(X \backslash a;t)\phi(X \backslash b;t) - \phi(X;t)\phi(X \backslash \{a,b\};t)} = \sum_{P \in \mathcal{P}_{ab}} \phi(X \backslash P;t).
\end{align}

Expressions \eqref{eq:neutrino1} and \eqref{eq:neutrino2} (or equivalent forms) have been used in various contexts for a long time, but they did not seem to be well known to the wide scientific community. They were rediscovered recently in the context of the physics of neutrino oscillations, leading to the vast survey \cite{NeutrinoIdentities} of their known uses, along with some media coverage.

\subsection{Variational principles for eigenvalues}

For the results in this subsection, we refer the reader to \cite[Chapter 3]{BhatiaMatrixAnalysis}.

Assume $A$ is a symmetric matrix acting on a finite vector space $\mathcal{V}$, and that $\lambda^\downarrow_k(A)$ denotes the $k$-th largest eigenvalue of $A$, and $\lambda^\uparrow_k(A)$ the $k$-th smallest. By $\mathcal{U} \subseteq \mathcal{V}$ we mean that $\mathcal{U}$ is a subspace of $\mathcal{V}$. The minimax principle for eigenvalues of symmetric matrices states that
\begin{align*}
	\lambda^\downarrow_k(A) & = \max_{\substack{\Uu \subseteq \Vv \\ \dim \Uu = k}} \min_{\substack{\ket v \in \Uu \\ \braket{v}{v} = 1}} \bra u A \ket u  = \min_{\substack{\Uu \subseteq \Vv \\ \dim \Uu = n-k+1}} \max_{\substack{\ket v \in \Uu \\ \braket{v}{v} = 1}} \bra u A \ket u.
\end{align*}

From this, several consequences ensue, and we list those which will be useful to us. The first is the well known Cauchy interlacing.

\begin{theorem}\label{thm:cauchy}
	Let $A$ be an $n\times n$ symmetric matrix, and let $S$ be an $n\times m$ matrix so that $S^T S = I$. Let $B = S^T A S$. Then
	\[
	\lambda^\downarrow_k(A) \geq \lambda^\downarrow_k(B) \quad \text{and} \quad \lambda^\uparrow_k(B) \geq \lambda^\uparrow_k(A).
	\]
\end{theorem}

Cauchy's interlacing says that the eigenvalues of a vertex-deleted subgraph lie in-between the eigenvalues of the original graph, thus, in particular, the multiplicity of an eigenvalue decreases by at most 1 upon the deletion of a vertex.

In our work, we will also need information about the eigenvalues of the sum of two symmetric matrices. The inequalities below are usually attributed to Weyl.

\begin{theorem}\label{thm:weyl}
	Let $A$ and $B$ be symmetric $n \times n$ matrices. Fix index $k$. Then, for all $i \leq k$,
	\[
		\lambda^\downarrow_k(A + B) \leq \lambda^\downarrow_i(A) + \lambda^\downarrow_{k-i+1}(B),	
	\]
	and, for all $i \geq k$,
	\[
		\lambda^\downarrow_k(A + B) \geq \lambda^\downarrow_i(A) + \lambda^\downarrow_{k-i+n}(B).
	\]
\end{theorem}

Finally, we will also require knowledge about the sum of eigenvalues of a matrix. The most general principle is usually known as Wielandt minimax which results in a theorem due to Lidskii, though we will only need the simpler form, shown below, an immediate consequence of a known result due to Ky Fan.

\begin{theorem}\label{thm:kyfan}
	Let $A$ and $B$ be symmetric $n\times n$ matrices. Then, for any $k \in \{1,\ldots,n\}$,
	\[
		\sum_{j = 1}^k \lambda^\downarrow_k(A + B) \leq \sum_{j = 1}^k \lambda^\downarrow_k(A) + \sum_{j = 1}^k \lambda^\downarrow_k(B)
	\]
\end{theorem}

\subsection{Double stars and extended double stars}

Our case analysis in the next sections will require us to rule out perfect state transfer in double stars and extended double stars. A star $S_k$ is the complete bipartite graph $K_{1,k}$, where $k$ is allowed to be $0$, in which case $S_0$ is the empty graph with one vertex.

If $Z$ is as in Figure \ref{figure1}, with $Y_1 = S_k$, $X=P_2$ and $Y_2 = S_\ell$, then $Z$ is a double star, denoted by $S_{k,\circ\circ,\ell}$. For these, the work is already done.

\begin{theorem}[\cite{FanGodsil}, Theorem 4.6] \label{2star}
There is no perfect state transfer on the double star graph $S_{k,\circ\circ,\ell}$ for $k$ or $\ell$ at least $1$.
\end{theorem}

If $Z$ is as in Figure \ref{figure1}, with $Y_1 = S_k$, $X=P_3$ and $Y_2 = S_\ell$, then $Z$ is an extended double star, denoted by $S_{k,\circ\circ\circ,\ell}$. As demonstrated by Hou, Gu, and Tong~\cite{HouGuTong}, these also do not admit perfect state transfer.

\begin{theorem} [\cite{HouGuTong}, Theorem 2.8] \label{ext2star}
There is no perfect state transfer on the extended double star graph $S_{k,\circ\circ\circ,\ell}$ for $k$ or $\ell$ at least $1$.
\end{theorem}

\section{Strong cospectrality for cut vertices} \label{sec:strongcut}

\textit{From this section on, we assume all polynomials use $t$ as as their variable. In order to simplify the notation, we will usually
denote the charcteristic polynomial of a graph $X$ by $\phi(X)$.}

\begin{theorem} \label{thm:walkequiv}
	Let $Z$ be given as in Figure \ref{figure1}. Assume $a$ and $b$ are cospectral in $X$. Thus, $a$ and $b$ are cospectral in $Z$ if and only if
	\[
		\frac{\phi(Y_1 \ba a)}{\phi(Y_1)} = \frac{\phi(Y_2 \ba b)}{\phi(Y_2)}.
	\]
\end{theorem}
\begin{proof}
	From the $1$-sum lemma (Lemma \ref{1sum}), it follows that:
\[
	\ch{Z \ba a} = \ch{Y_1 \ba a} \cdot (\ch{X\ba b} \ch{Y_2} + \ch{X} \ch{Y_2 \ba b} - t \ch{X \ba b}\ch{Y_2 \ba b})
\]
\[
	\ch{Z \ba b} = \ch{Y_2 \ba b} \cdot (\ch{X \ba a} \ch{Y_1} + \ch{X} \ch{Y_1 \ba a} - t \ch{X \ba a}\ch{Y_1 \ba a})
\]
	Note that $\ch{X \ba a} = \ch{X\ba b}$, as a consequence of Equation \eqref{eq:adju}, as $a$ and $b$ are cospectral in $X$. It follows that $\ch{Z \ba a} = \ch{Z \ba b}$ if and only if
\[
    \ch{Y_1 \ba a}\ch{Y_2} = \ch{Y_2 \ba b}\ch{Y_1}. \qedhere
\]    
\end{proof}

We will say that vertices $a \in Y_1$ and $b \in Y_2$ are walk equivalent if they satisfy the condition in the previous theorem. 

Recall from Theorem \ref{thm:pstcha} that we require $E_r \ket a = \pm E_r \ket b$ in order for perfect state transfer to hold (meaning, that $a$ and $b$ are strongly cospectral). The result above provides a condition for $a$ and $b$ to be cospectral. Fortunately, when there is a unique path joining $a$ and $b$, we can show that the two are equivalent.

For the result below, we use Lemma 2.4 from \cite{CoutinhoGodsilPSTpolytime} that says that $a$ and $b$ are strongly cospectral in a given graph $X$ if and only if $\phi(X\ba a) = \phi(X \ba b)$ and the poles of $\phi(X \ba ab)/\phi(X)$ are simple.

\begin{theorem} \label{thm:path}
	Let $Z$ be a graph as in Figure \ref{figure1}. Assume the graph $X$ is a path (and thus $a$ and $b$ are cospectral in $X$). Then they are cospectral in $Z$ if and only if they are strongly cospectral in $Z$.
\end{theorem}

\begin{proof}
	The only thing we need to show is that the poles of $\phi(Z \ba ab)/\phi(Z)$ are simple.
	
	From Equations \eqref{eq:neutrino0} and \eqref{eq:neutrino2}, we have that
	\[
		(tI-A(W))^{-1}_{a,b} = \frac{\sqrt{\phi(W \backslash a)\phi(Z \backslash b) - \phi(Z)\phi(Z \backslash ab)}}{\phi(Z)}
	\]
	has simple poles (and this is also true with $X$ instead of $Z$).

	From Equation \eqref{eq:neutrino3}, it follows that
	\[
		\sqrt{\phi(Z \backslash a)\phi(Z \backslash b) - \phi(Z)\phi(Z \backslash ab)} = \phi(Y_1 \ba a) \phi(Y_2 \ba b) \phi(X \ba P) = \phi(Y_1 \ba a) \phi(Y_2 \ba b).
	\]
	
	Finally, note that
	\begin{align*}
		\frac{\phi(Z \ba ab)}{\phi(Z)} & = \frac{\phi(Y_1 \ba a) \phi(Y_2 \ba b) \phi(X \ba ab)}{\phi(Z)} \\
		& = \frac{\sqrt{\phi(Z \backslash a)\phi(Z \backslash b) - \phi(Z)\phi(Z \backslash ab)}\phi(X \ba ab)}{\phi(Z)},
	\end{align*}
	which has simple poles.
\end{proof}

\section{No state transfer over one bridge} \label{sec:p2}

In this section, we will show that if two vertices are joined by a bridge, then there is no perfect state transfer between them (unless the graph itself is $P_2$).

\begin{theorem} \label{thm:walk-equiv-loops}
	Let $Z$ be given as in Figure \ref{figure1}, and assume $X = P_2$. Assume $a$ and $b$ are strongly cospectral in $Z$. The following are equivalent.
	\begin{enumerate}[(a)]
		\item $\theta$ is eigenvalue of $A(Y_1) + \ketbra{a}{a}$ in the support of $a$ 
		\item $\theta$ is eigenvalue of $A(Y_2) + \ketbra{b}{b}$ in the support of $b$ 
		\item $\theta$ is eigenvalue of $A(Z)$ with $E_\theta \ket a = E_\theta \ket b \neq 0$.
	\end{enumerate}
	The following are equivalent.
	\begin{enumerate}[(a)]
		\item $\theta$ is eigenvalue of $A(Y_1) - \ketbra{a}{a}$ in the support of $a$ 
		\item $\theta$ is eigenvalue of $A(Y_2) - \ketbra{b}{b}$ in the support of $b$ 
		\item $\theta$ is eigenvalue of $A(Z)$ with $E_\theta \ket a = - E_\theta \ket b \neq 0$.
	\end{enumerate}
	Moreover, the eigenvalues of $A(Z)$ not in the support of $a$ and $b$ are eigenvalues of $A(Y_1) \pm \ketbra{a}{a}$ not in the support of $a$ or of $A(Y_2) \pm \ketbra{b}{b}$ not in the support of $b$.
\end{theorem}

\begin{proof}
	First, to see how eigenvalues of $Z$ relate to eigenvalues of $A(Y_1) \pm \ketbra{a}{a}$ and of $A(Y_2) \pm \ketbra{b}{b}$, it is sufficient to think in terms of projecting eigenvectors. For instance, assume $\theta$ is eigenvalue of $Z$ in the support of $a$, with $E_\theta \ket a = E_\theta \ket b$, and let $f : V(Z) \to \Rds$ be a corresponding eigenvector. Then
	\[
		\theta f(a) = \sum_{u \sim a} f(u) \implies \theta f(a) = f(a) + \sum_{u \sim a,\ u \neq b} f(u)
	\]
	Then it is immediate to verify that $\theta$ is a root of $\ch{Y_1} + \ch{Y_1 \ba a}$ in the support of $a$, and of $\ch{Y_2} + \ch{Y_2 \ba b}$ in the support of $b$. Note that these are the characteristic polynomials of the graphs $Y_1$ and $Y_2$ with a loop of weight $+1$ added at vertices $a$ and $b$ respectively.
	
	Likewise, if $\theta$ is eigenvalue of $Z$ with $E_\theta \ket a = - E_\theta \ket b \neq 0$, then $\theta$ is a root of $\ch{Y_1} - \ch{Y_1 \ba a}$ and of $\ch{Y_2} - \ch{Y_2 \ba b}$.
	
	Finally, if 	$\theta$ is eigenvalue of $Z$ not in the support of $a$ and $b$, then it is an eigenvalue of both of the graphs $Y_1$ and $Y_1 \ba a$ or of both of the graphs $Y_2$ and $Y_2 \ba b$. 
	
	Second, we now relate eigenvalues of $A(Y_1) \pm \ketbra{a}{a}$ and of $A(Y_2) \pm \ketbra{b}{b}$ to eigenvalues of $Z$. From applying the $1$-sum lemma (Lemma \ref{1sum}) twice, we get
	\[
		\ch{Z} = \ch{Y_1}\ch{Y_2} - \ch{Y_1 \ba a} \ch{Y_2 \ba b}.
	\]
	Thus, because $a$ and $b$ are walk equivalent (Theorem \ref{thm:walkequiv}),
	\[
		\ch{Z} = (\ch{Y_1} \pm \ch{Y_1 \ba a}) \ (\ch{Y_2} \mp \ch{Y_2 \ba b}).
	\]
	Thus, if $\theta$ is root of $(\ch{Y_1} + \ch{Y_1 \ba a})$, then it is also of $\ch{Z}$. If $\theta$ is in the support of $a$ in $A(Y_1) + \ketbra{a}{a}$, then Equation \ref{eq:neutrino1} implies
	\[
		\frac{\ch{Y_1 \ba a} (t - \theta)}{\ch{Y_1} + \ch{Y_1 \ba a}} \bigg|_{t = \theta} \neq 0.
	\]
	From interlacing (Theorem \ref{thm:cauchy}), we have that the multiplicity of $\theta$ in $\ch{Y_1 \ba a}$ is exactly one unity smaller than its multiplicity in $(\ch{Y_1} + \ch{Y_1 \ba a})$, hence its multiplicity in $\ch{Y_1}$ is equal to its multiplicity in $\ch{Y_1 \ba a}$. Moreover,
	\[
		\ch{Z\ba a} = \ch{Y_1 \ba a}\ch{Y_2},
	\]
	and from the walk equivalence,
	\[
		\frac{\ch{Y_2}}{\ch{Y_2}-\ch{Y_2 \ba b}} = \frac{\ch{Y_1}}{\ch{Y_1}-\ch{Y_1 \ba a}}.
	\]
	Piecing everything together, we can conclude that
	\[
		\frac{\ch{Z \ba a} (t - \theta)}{\ch{Z}} \bigg|_{t = \theta} \neq 0,
	\]
	therefore $\theta$ is in the support of $a$ in $Z$.
	
	An analogous argument holds for when $\theta$ is eigenvalue of $A(Y_1) - \ketbra{a}{a}$ in the support of $a$ or of $A(Y_2) \pm \ketbra{b}{b}$ in the support of $b$.
\end{proof}

\begin{theorem} \label{thm:nopstbridge}
	Let $Z$ be given as in Figure \ref{figure1}, with $X = P_2$. If there is perfect state transfer between $a$ and $b$, then the graphs $Y_1$ and $Y_2$ have only one vertex each.
\end{theorem}
\begin{proof}
	Vertices $a$ and $b$ are strongly cospectral. Let $\Phi^\pm_{ab}$ be the eigenvalues $\theta$ in the support of these vertices so that $E_\theta \ket a = \pm E_\theta \ket b$.
		
	Let $M$ be a matrix that represents the action of $A(Y_1)$ in an orthogonal basis that contains $\ket a$ for the walk module generated by $\ket a$ in $\Rds^{V(Y_1)}$. If this module has dimension $m$, let $E_0$ be the $m \times m$ matrix with $1$ in its first position, and $0$s elsewhere. It is immediate to verify that $M \pm E_0$ represents the action of $A(Y_1) \pm \ketbra{a}{a}$ on the walk module generated by $\ket a$, according to the same basis.
	
	From Theorem \ref{thm:weyl}, it follows that
	\[
		\lambda^\downarrow_j(M + E_0) \geq \lambda^\downarrow_j(M - E_0) + \lambda^\downarrow_m(2 E_0) = \lambda^\downarrow_j(M - E_0).
	\]
	Let $s$ be the sum of the eigenvalues of $A(Y_1) \pm \ketbra{a}{a}$ outside of the support of $a$. It is a consequence of Theorem \ref{thm:walk-equiv-loops} that  $\Phi^\pm_{ab}$ are the eigenvalues of $M \pm E_0$, and using the inequality above, the fact that the sets $\Phi^+_{ab}$ and $\Phi^-_{ab}$ are disjoint, and also that all distinct eigenvalues in the support of $a$ and $b$ differ by at least $1$ (Theorem \ref{thm:pstcha}, item b), we have that
	\begin{align*}
		1 & = \tr (A(Y_1) +  \ketbra{a}{a} ) \\ & = s + \sum_{\theta \in \Phi^+_{ab}} \theta \\ & \geq s + \sum_{\theta \in \Phi^-_{ab}} (\theta + 1) \\ & = m + \tr (A(Y_1) -  \ketbra{a}{a}) \\ & = m - 1.
	\end{align*}
	
	Hence $m \leq 2.$ 
	
	If equality holds we have $\Phi^+_{ab} = \{\theta_1, \theta_2\}$ and $\Phi^-_{ab} = \{\theta_1 - 1, \theta_2 - 1\}$.  As the dimension of the walk module of $\ket a$ in $Y_1$ is $2$, its covering radius is at most 1, and thus $a$ is a universal vertex (meaning, its a neighbour to all vertices in $Y_1 \ba a$).  

	 Now, there exists an eigenbasis of $A(Y_1)$ such that $|V(Y_1)| - 2$ of the vectors $\ket x$ in the basis are such that $\braket{a}{x} = 0$ (because there are only two distinct eigenvalues in the support of $a$). It follows that these vectors $\ket x$ sum to $0$ in the neighbourhood of $a$, which is $Y_1 \ba a$, and therefore $\braket{x}{\1} = 0$. The restriction of these vectors to $Y_1 \ba a$ are also eigenvectors of $Y_1\ba a$, and this graph has precisely $|V(Y_1)| - 1$ linearly independent eigenvectors.  Thus, the remaining eigenvector of $Y_1 \ba a$ is $\mathbf{1}$, so $Y_1 \ba a$ is regular; we assume of degree $k$.

	It follows that if $n = |V(Y \ba a)|$, then $\theta_1, \theta_2$ are eigenvalues of the quotient matrix 
	\[
	\begin{bmatrix} 1 & \sqrt{n} \\ \sqrt{n} & k \end{bmatrix}
	\]
	and $\theta_1 - 1, \theta_2 - 1$ are eigenvalues of the quotient matrix
	\[
	\begin{bmatrix} -1 & \sqrt{n} \\ \sqrt{n} & k \end{bmatrix}.
	\]
	Hence, we have
	\[ \theta_1 \theta_2 = k - n ,\quad \theta_1 + \theta_2 = k + 1 ,\quad \text{and}\quad  (\theta_1 - 1) (\theta_2 - 1) = -k - n\]
	which imply $k = 0$, and thus $Y_1 \ba a = \overline{K}_n$.

	Therefore $Z$ is a double star, and these do not admit perfect state transfer according to Theorem~\ref{2star}.

	The only case left is $m = 1$, so $Y_1 = K_1$, and by a symmetric argument $Y_2 = K_1$, as we wanted.
\end{proof}

\section{No state transfer over two bridges} \label{sec:p3}

Assuming the graph $Z$ given as in Figure \ref{figure2}, and assume that $X = P_3$. Define graphs $Z_1$ and $Z_2$, as in Figures \ref{figure3} and \ref{figure4}:

\begin{figure}[H]
    \centering
	\includegraphics[scale=0.5]{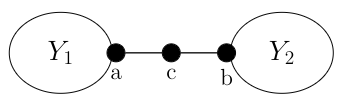}
    \caption{Graph $Z$} \label{figure2}
\end{figure}

\begin{figure}[H]
\centering
\begin{minipage}{.5\textwidth}
  \centering
  \includegraphics[scale=0.5]{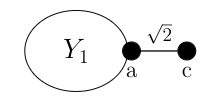}
  \caption{Graph $Z_1$} \label{figure3}
\end{minipage}%
\begin{minipage}{.5\textwidth}
  \centering
  \includegraphics[scale=0.5]{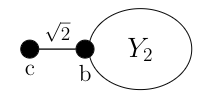}
  \caption{Graph $Z_2$} \label{figure4} 
\end{minipage}
\end{figure}

\begin{theorem} \label{thm:walk-equiv-wpath}
	Let $Z$, $Z_1$ and $Z_2$ be as in Figures \ref{figure2}, \ref{figure3}, and \ref{figure4}. Assume $a$ and $b$ are strongly cospectral in $Z$. 	The following are equivalent.
	\begin{enumerate}[(a)]
		\item $\theta$ is eigenvalue of $A(Z_1)$ in the support of $a$ 
		\item $\theta$ is eigenvalue of $A(Z_2)$ in the support of $b$ 
		\item $\theta$ is eigenvalue of $A(Z)$ with $E_\theta \ket a = + E_\theta \ket b \neq 0$.
	\end{enumerate}
	The following are equivalent.
	\begin{enumerate}[(a)]
		\item $\theta$ is eigenvalue of $A(Y_1)$ in the support of $a$ 
		\item $\theta$ is eigenvalue of $A(Y_2)$ in the support of $b$ 
		\item $\theta$ is eigenvalue of $A(Z)$ with $E_\theta \ket a = - E_\theta \ket b \neq 0$.
	\end{enumerate}
	Moreover, the eigenvalues of $A(Z)$ not in the support of $a$ and $b$ are eigenvalues of $A(Y_1)$ not in the support of $a$ or of $A(Y_2)$ not in the support of $b$, or possibly the eigenvalue $0$ if it is an eigenvalue of $A(Z_1)$ or $A(Z_2)$.
\end{theorem}

\begin{proof}
	From applying the $1$-sum lemma (Lemma \ref{1sum}) twice, we get
	\[
		\ch{Z} = t \ch{Y_1}\ch{Y_2} - \ch{Y_2}\ch{Y_1 \ba a} - \ch{Y_1}\ch{Y_2 \ba b}.
	\]
	Thus, because $a$ and $b$ are walk equivalent (Theorem \ref{thm:walkequiv}),
	\begin{align}\label{eq:1}
		\ch{Z} = \ch{Y_1}(t \ch{Y_2} - 2 \ch{Y_2 \ba b}) = \ch{Y_2}(t \ch{Y_1} - 2 \ch{Y_1 \ba a}).
	\end{align}
	From this, it follows that eigenvalues of $Z$ are either eigenvalues of $Y_1$ or $Z_2$ (and equivalently either of $Y_2$ or $Z_1$). Let us now check the correspondence between the eigenvalue supports of $a$ and $b$.
	
	Assume $\theta$ is eigenvalue of $Z$ in the support of $a$, with $E_\theta \ket a = E_\theta \ket b$, and let $f : V(Z) \to \Rds$ be a corresponding eigenvector. Then
	\[
		\theta f(a) = \sum_{u \sim a} f(u) \implies \theta f(a) = \frac{f(a) + f(b)}{\theta} + \sum_{u \sim a,\ u \neq c} f(u)
	\]
	Then it is immediate to verify that $\theta$ is a root of $A(Z_1)$ in the support of $a$, as $f(a) = f(b)$, and also a root of $A(Z_2)$ in the support of $b$. Note that it also follows that $\theta \neq 0$.
	
	Likewise, if $\theta$ is eigenvalue of $Z$ with $E_\theta \ket a = - E_\theta \ket b \neq 0$, then any $\theta$ eigenvector sums to $0$ on the neighbours of $c$, and thus either $\theta$ is eigenvalue of both $Y_1$ and $Y_2$, or $\theta = 0$, but in this latter case \eqref{eq:1} implies that $\theta = 0$ is eigenvalue for $Y_1$ and $Y_2$.
		
	Finally, if 	$\theta$ is eigenvalue of $Z$ not in the support of $a$ and $b$, then it is an eigenvalue of both of the graphs $Y_1$ and $Z_1$ or of both of the graphs $Y_2$ and $Z_2$. 
	
	For the converse direction, first recall Equation \eqref{eq:neutrino1}. We note that an eigenvalue $\theta$ of $Z$ is in the support of $a$ if and only if
	\begin{align} \label{eq:2}
		\frac{\ch{Z \ba a}(t-\theta)}{\ch{Z}}\bigg|_{t= \theta} & = \frac{\ch{Y_1 \ba a} (t\ch{Y_2} - \ch{Y_2 \ba b}) (t-\theta)}{\ch{Y_1} (t\ch{Y_2} - 2 \ch{Y_2 \ba b})} \bigg|_{t= \theta} \notag \\ & = \frac{\ch{Y_2 \ba b} (t\ch{Y_2} - \ch{Y_2 \ba b}) (t-\theta)}{\ch{Y_2} (t\ch{Y_2} - 2 \ch{Y_2 \ba b})} \bigg|_{t= \theta} \neq 0. 
	\end{align}
	If $\theta$ is eigenvalue of $Z_2$ in the support of $b$, then
	\begin{align} \label{eq:3}
		\frac{t \ch{Y_2 \ba b} (t - \theta)}{t \ch{Y_2} - 2\ch{Y_2 \ba b}} \bigg|_{t = \theta} \neq 0,
	\end{align}
	but also recall that $\theta \neq 0$ and $(t \ch{Y_2} - 2\ch{Y_2 \ba b}) = 0$. If both terms are non-zero at $t = \theta$, then \eqref{eq:2} clearly holds. If $\ch{Y_2}(\theta) = \ch{Y_2 \ba b}(\theta) = 0$, then \eqref{eq:3} implies the multiplicity in $\ch{Y_2}$ is one larger than that in $\ch{Y_2 \ba b}$, and this ensures \eqref{eq:2} holds. Therefore, because $a$ and $b$ are strongly cospectral in $Z$, we have that $\theta$ is in the support of $b$ in $Z$. An analogous argument holds with the roles of $a$ and $b$ reversed.
	
	If $\theta$ is eigenvalue of $Y_2$ in the support of $b$, then
	\begin{align} \label{eq:4}
		\frac{\ch{Y_2 \ba b} (t - \theta)}{\ch{Y_2}} \bigg|_{t = \theta} \neq 0,
	\end{align}
	and interlacing implies that the multiplicity of $\theta$ in $\ch{Y_2 \ba b}$ in one unity smaller than in $\ch{Y_2}$. This gives \eqref{eq:2} immediately, and $\theta$ is in the support of $b$ in $Z$. An analogous argument holds with the roles of $a$ and $b$ reversed.
\end{proof}

\begin{theorem}
	Let $Z$ be as in Figure \ref{figure2}. If there is perfect state transfer between $a$ and $b$, then the graphs $Y_1$ and $Y_2$ have one vertex each.
\end{theorem}
\begin{proof}
	Assume $a$ and $b$ are strongly cospectral, and let $\Phi^\pm_{ab}$ be the eigenvalues $\theta$ in the support of these vertices so that $E_\theta \ket a = \pm E_\theta \ket b$.
		
	Let $M$ be a matrix that represents the action of $A(Y_1)$ in an orthogonal basis that contains $\ket a$ for the walk module generated by $\ket a$ in $\Rds^{V(Y_1)}$. If this module has dimension $m$, let $E_0$ be the $(m+1) \times (m+1)$ matrix with $0$s in all positions, except for its $(1,2)$ and $(2,1)$ entries, both equal to $\sqrt{2}$. Also, pad $M$ with a first row and first column both equal to $0$, call this $M'$. It is immediate to verify that $M' + E_0$ represents the action of $A(Z_1)$ in the walk module generated by $\ket c$ in $\Rds^{V(Z_1)}$. Note that the walk module generated by $\ket a$ is contained in this one, and they are different if and only if $0$ is an eigenvalue of $A(Z_1)$ in the support of $c$ but not in the support of $a$. Also note that $0$ is never an eigenvalue of $A(Z_1)$ in the support of $a$. As a consequence, the non-zero eigenvalues of $M' + E_0$ are precisely the eigenvalues of $A(Z)$ in $\Phi_{ab}^+$ (as per Theorem \ref{thm:walk-equiv-wpath}).
	
	From interlacing (Theorem \ref{thm:cauchy}), it follows that, for all $j$,
	\[
		\lambda^\downarrow_j(M' + E_0) \geq \lambda^\downarrow_j(M) \quad \text{and} \quad \lambda^\uparrow_j(M' + E_0) \leq \lambda^\uparrow_j(M).
	\]
	
	We consider then two cases below. For both, recall that Theorem \ref{thm:walk-equiv-wpath} establishes that the eigenvalues of $A(Z_1)$ in the support of $a$ and those of $A(Y_1)$ in the support of $a$ are the eigenvalues in the support of $a$ in $Z$, and from Theorem \ref{thm:pstcha}, item b, we have that distinct eigenvalues in this set differ by at least $1$. Also recall that eigenvalues of $(M'+E_0)$ and of $M$ are simple.
	
	\begin{enumerate}[(i)]
		\item $0$ is an eigenvalue of $M' + E_0$. In this case, assume $(M'+E_0)$ has two positive eigenvalues. Then $M$ has two non-negative eigenvalues (from interlacing), and therefore we can assume that $\lambda^\downarrow_1(M')$ and $\lambda^\downarrow_2(M')$ are eigenvalues of $M$. Thus, from interlacing, we have
	\[
    	\sum_{j = 1}^{2} \lambda^\downarrow_{j}(M'+E_0) \geq \sum_{j = 1}^{2} (\lambda^\downarrow_{j}(M) + 1) > \sum_{j = 1}^{2} \lambda^\downarrow_{j}(M) + \sqrt{2} = 
                \sum_{j = 1}^{2} \lambda^\downarrow_{j}(M') + \sum_{j = 1}^{2} \lambda^\downarrow_{j}(E_0),
    \]
	which contradicts Theorem \ref{thm:kyfan}. A similar argument also shows that $(M'+E_0)$ does not have at least two negative eigenvalues.
	
		\item $0$ is not an eigenvalue of $M' + E_0$. In this case, assume $M$ has at least two non-negative eigenvalues, and, thus, from interlacing, $M' + E_0$ has two positive eigenvalues. An argument similar to the one above arrives at a contradiction. Thus in this case, $M$ can only have one non-negative eigenvalue and one non-positive eigenvalue.
	\end{enumerate}
	
	In summary, either $0$ is an eigenvalue of $M' + E_0$ and $M' + E_0$ has at most three distinct eigenvalues, or $0$ is not an eigenvalue of $M' + E_0$ and $M$ has at most two distinct eigenvalues. In either case, we conclude that there at most two distinct eigenvalues in the support of $a$ either in $Z_1$ or in $Y_1$ respectively, and therefore $a$ must be a neighbour to all vertices in $Y_1$. 
	
	For the first case, there exists an eigenbasis of $A(Z_1)$ such that $|V(Z_1)| - 2$ of the vectors $\ket x$ in this basis are such that $\braket{a}{x} = 0$.  It follows that these vectors sum to $0$ in the neighbourhood of $a$, and therefore $\braket{x}{\1'} = 0$, where $\1'$ has all entries equal to $1$ but for the entry corresponding to $c$, which is equal to $\sqrt{2}$. The restriction of these vectors to $Z_1 \ba a$ are eigenvectors of $Z_1 \ba a$, thus the remaining eigenvector of $Z_1\ba a$ is $\mathbf{1}'$, and this immediately implies that $Z_1 \ba a$ is regular of degree $0$.
	
	For the second case, a similar argument to the one above (also similar to the argument in the proof of Theorem \ref{thm:nopstbridge}) shows that $Y_1 \ba a$ is regular of degree $k$ (we cannot immediately give that $k = 0$, but this is the case, as we show below).
	
	Let $\theta^+, \theta^-$ be the two eigenvalues in the support of $a$ in $A(Y_1)$, and let $\lambda^+$, $\lambda^0$ and $\lambda^-$ be the eigenvalues in the support of $a$ in $A(Z_1)$.	It follows that if $n = |V(Y_1 \ba a)|$, then $\theta^+, \theta^-$ are eigenvalues of the quotient matrix
    	
    \[
    \ov{A(Y_1)} = \begin{bmatrix} 0 & \sqrt{n} \\ \sqrt{n} & k \end{bmatrix}.
    \]
    	
    and $\lambda^+$, $\lambda^0$ and $\lambda^-$ are eigenvalues of the quotient matrix
    	
    \[
    \ov{A(Z_1)} = \begin{bmatrix} 0 & \sqrt{2} & 0 \\ \sqrt{2} & 0 & \sqrt{n} \\ 0 & \sqrt{n} & k \end{bmatrix}.
    \]
    
    It follows from Theorem \ref{thm:weyl} that
    \begin{align*}
    \lambda^+ &\leq \theta^+ + \sqrt{2}, \text { and}\\
    \lambda^- &\geq \theta^- - \sqrt{2}.
    \end{align*}
    
    From interlacing and from Theorem \ref{thm:pstcha}, we know that ${\lambda^+ > \theta^+ > \lambda^0 > \theta^- > \lambda^-}$, and each inequality holds by least a multiple of $\sqrt{\Delta}$. Thus $\Delta \in \{1, 2\}$, and
    
    \begin{align*}
    \lambda^+ & = \theta^+ + \sqrt{\Delta} , \text{ and }\\
    \lambda^- & = \theta^- - \sqrt{\Delta}.
    \end{align*}
    
    Calculating the trace of both matrices, we get that
    \begin{align*}
    \theta^+ + \theta^- & = k , \text{ and }\\
    \lambda^+ + \lambda^0 + \lambda^- = \theta^+ + \sqrt{\Delta} + \lambda^0 + \theta^- - \sqrt{\Delta} & = k.
    \end{align*}
    
    Thus $\lambda^0 = 0$, but the free term of the characteristic polynomial of $\ov{A(Z_1)}$ is $-2k$, thus $0$ is an eigenvalue if and only if $k = 0$, therefore $Y_1 \ba a = \overline{K_n}$.
    
    If $Y_1 \ba a$ and $Y_2 \ba b$ are non-empty, then $Z$ is an extended double star, and these do not admit perfect state transfer according to Theorem~\ref{ext2star}.

	The only case left is when $Y_1 = \{a\}$ and $Y_2 = \{b\}$, as we wanted.
\end{proof}

\section{Conclusion} \label{sec:final}

One main motivation of this paper is Conjecture 1 in \cite{CoutinhoLiu2} that proposes that $P_2$ and $P_3$ are the only trees admitting perfect state transfer. We were able to show in this paper that if perfect state transfer happens between $a$ and $b$ in the graph $Z$ (as in Figure \ref{figure1}) for when $X = P_2, P_3$, then $Z = P_2, P_3$ respectively. Note that extending this result to show a no-go theorem for perfect state transfer between a vertex in $Y_1$ to a vertex in $Y_2$ would imply the no state transfer in trees conjecture. We are not ready to state this extension as a conjecture, but we list it as an open problem.

\begin{problem}
	Consider $Z$ as in Figure \ref{figure1}, have $X = P_2$, and assume $Y_1$ and $Y_2$ have at least two vertices. Find an example of such $Z$ admitting perfect state transfer between a vertex in $Y_1$ to a vertex in $Y_2$, or show that none exists.
\end{problem}

Another natural extension of our work in this paper consists in determining for which other graphs $X$ an analogous result holds. We believe that the result is true for when $X$ is a longer path, but a naive attempt in finding a inductive proof did not succeed. We now assume the graph $Z$ looks like the figure below.
\begin{figure}[h!]
    \centering
	\includegraphics[scale=0.5]{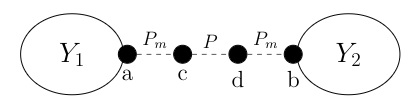}
    \caption{Graph $Z$} \label{figure6}
\end{figure}

We can show that if $a$ and $b$ are strongly cospectral, then so are $c$ and $d$, but we cannot guarantee that if perfect state transfer occurs between $a$ and $b$, then it also does between $c$ and $d$, because these latter vertices could have other eigenvalues in their support which are not in the supports of $a$ and $b$. 

An alternative approach could be to generalize the application of the 1-sum lemma to this case, but this does not seem too promising.

\begin{conjecture}
	Consider $Z$ as in Figure \ref{figure6}. Perfect state transfer does not occur between $a$ and $b$.
\end{conjecture}

A third and last problem we propose is that of characterizing when cut vertices are strongly cospectral. We have shown in Theorem \ref{thm:walkequiv} that if $a$ and $b$ are cospectral in $X$, they are cospectral in $Z$ depending only on the graphs $Y_1$ and $Y_2$, and Theorem \ref{thm:path} shows a condition for this cospectrality to become strong. This leads to two problems:

\begin{problem}
	Consider $Z$ as in Figure \ref{figure1}, $a$ and $b$ cospectral in both $X$ and $Z$. 	What (natural) condition on the graph $X$ is equivalent to $a$ and $b$ becoming strongly cospectral in $Z$? Theorem \ref{thm:path} shows that $X$ itself being a path is sufficient, but this is certainly not necessary. We warn though that $a$ and $b$ being strongly cospectral in $X$ or for it to be a unique path between $a$ and $b$ are both not enough conditions.    
\end{problem}

\begin{problem}
	Find a general construction of graphs as in Figure \ref{figure1} so that $a$ and $b$ are strongly cospectral in $Z$ but not even cospectral in $X$. We have at least one example, but we do not know how to generalize it. 
\end{problem}

\section*{Acknowledgements}

E. Juliano acknowledges grant PROBIC/FAPEMIG. C. Godsil gratefully ac-
knowledges the support of the Natural Sciences and Engineering Council of Canada (NSERC), Grant No.RGPIN-9439. C.M. van Bommel acknowledges PIMS Postdoctoral Fellowship.

\printbibliography


\end{document}